\documentclass[10pt,conference]{IEEEtran}

\usepackage{verbatim}
\usepackage{url}
\usepackage{amsmath}
\usepackage{amsfonts}
\usepackage{amssymb}
\usepackage{setspace}
\usepackage{paralist}
\usepackage{enumerate}
\usepackage[pdftex]{graphicx} 
\usepackage[pdftex,colorlinks,
           bookmarks=true,%
           pdftitle=Analysis\ of\ a\ Proportionally\ Fair\ and\ Locally\ Adaptive\ spatial\ Aloha\ in\ Poisson\ Networks,%
           pdfauthor=F.\ Baccelli\ -\  B.\ Blaszczyszyn\ -\ C.\ Singh,%
           ]{hyperref}  
\usepackage[numbers,sort&compress]{natbib} 
\usepackage{hypernat} 
\hypersetup{
   bookmarksnumbered,
   pdfstartview={FitH},
   citecolor={blue},
   linkcolor={red},
   urlcolor=[rgb]{0,0.55,0},
   pdfpagemode={UseOutlines}}

\usepackage[scriptsize,it]{caption}   
\newtheorem{Theorem}{Theorem}[section]
\newenvironment{theorem}{\bf\begin{Theorem}\rm\em}{\end{Theorem}} 
\newtheorem{Lemma}{Lemma}[section]
\newenvironment{lemma}{\bf\begin{Lemma}\rm\em}{\end{Lemma}} 
\newtheorem{remark}{Remark}[section]
\newtheorem{Proposition}{Proposition}[section]
 \newtheorem{Corollary}{Corollary}[section]
\newenvironment{corollary}{\bf\begin{Corollary}\rm\em}{\end{Corollary}} 

\newcommand{\R}{\mathbb{R}}
\newcommand{\one}{\mathbf{1}}
\newcommand{\e}{{\rm e}}
\newcommand{\sinr}{\textrm{SINR}}
\newcommand{\calF}{\mathcal{F}}

\newcommand{\maximize}{\operatornamewithlimits{maximize~~~}}
\newcommand{\inargument}{\operatornamewithlimits{in~~}}

\newcommand{\subjectto}{\operatornamewithlimits{subject~to~~~}}
\newcommand{\remove}[1]{}

\graphicspath{{./}{./plots/}}

\newcommand{\calS}{\mathcal{S}}
\renewcommand{\Pr}{\mathbb{P}}
\newcommand{\E}{\mathbb{E}}

\allowdisplaybreaks
\begin{document}
\title{Analysis of a Proportionally Fair and Locally Adaptive
Spatial Aloha in Poisson Networks
\thanks{
 This work was carried out in parts at Laboratory of Information, Networking and Communication Sciences~(LINCS) Paris and was supported by the INRIA-Alcatel Lucent Bell Labs Joint Research Center.  }}
\author{
\IEEEauthorblockN{Fran{\c c}ois Baccelli}
\IEEEauthorblockA{
University of Texas, Austin, USA \\
and  Inria--ENS, France
 \\
 Email: francois.baccelli@austin.utexas.edu\\[-4ex]}
    \and
\IEEEauthorblockN{Bart{\l}omiej B{\l}aszczyszyn}
\IEEEauthorblockA{
Inria--ENS, France\\
23 avenue d'Italie  75214 Paris\\
 Email: Bartek.Blaszczyszyn@ens.fr\\[-4ex]}
   \and
 \IEEEauthorblockN{Chandramani Singh}
 \IEEEauthorblockA{
  Coordinated Science Lab \\
UIUC, Urbana, IL 61801, USA\\
  Email: chandramani.singh@inria.fr\\[-4ex]} 
 }
\maketitle

\begin{abstract}
The proportionally fair sharing of the capacity of a Poisson
network using Spatial-Aloha leads to closed-form performance
expressions in two extreme cases: (1) the case without topology
information, where the analysis boils down to a parametric
optimization problem leveraging stochastic geometry;
(2) the case with full network topology information,
which was recently solved using shot-noise techniques.
We show that there exists a continuum of adaptive controls between
these two extremes, based on local stopping sets,
which can also be analyzed in closed form. We also show that
these control schemes are implementable, in contrast to the
full information case which is not. As local information increases,
the performance levels of these schemes are shown to get arbitrarily
close to those of the full information scheme. The analytical
results are combined with discrete event simulation to provide
a detailed evaluation of the performance of this class of medium
access controls.
\end{abstract}
\section{Introduction}
This work is focused on a version of 
Spatial Aloha~\cite{ctrltheory-wireless.baccelli-etal06aloha-multihop-wireless} where each node
willingly controls his medium access probability (MAP) in order to maximize
the network-wide sum of the logarithms of the node throughputs,
a global objective referred to as proportional fairness in the literature.
A good analogy is that of TCP users who willingly
throttle their transmission rates so as to maximize some network wide utility \cite{Kelly}.
In TCP, the transmission control is both {\em adaptive and decentralized}:
it is decentralized since each user throttles his rate based on his own packet losses,
which are seen as an indicator of the presence of contending users;
it is adaptive in that the throttling is substantial if there are many contenders (and hence losses)
at any given time and small otherwise.
In the version of Spatial Aloha studied in this paper, the access
control is also decentralized and adaptive: it is decentralized because each node
computes its MAP using an
equation based on some {\em local spatial information} subsuming its wireless contention status;
it is adaptive w.r.t. local spatial conditions in that a node will apply 
an heavy throttling if there are many wireless receivers nearby, and a small
one (or none) otherwise.

Stochastic geometry \cite{stochproc-wireless.baccelli-blaszczyszyn09stochastic-geometry-wireless-networks-1}
has recently been used for the analysis
and performance evaluation of wireless~(ad hoc as well as cellular) networks;
in this approach, one models node locations as a spatial point process,
e.g., a homogeneous Poisson point process~(PPP), and one
computes various network statistics, e.g., the distribution of interference,
the successful transmission probability, the coverage~(or outage)
probability etc. as spatial averages. This leads
to closed form expressions for a variety of performance metrics
that are then amenable to optimization
with respect to network parameters 
(node density, protocol parameters, etc.)~\cite{stochproc-wireless.baccelli-blaszczyszyn09stochastic-geometry-wireless-networks-2}.
In the Aloha case,
this approach takes a macroscopic view of the network with
the underlying assumption, justified by homogeneity, that all nodes
in the network have identical statistical characteristics.
The parameter to be optimized is the MAP, and at the optimum
operation point, all users have the same MAP as a direct
consequence of homogeneity.

Analyzing the behavior of the adaptive version of Spatial
Aloha described above in the context of large random networks
requires extending stochastic geometry to the situation where
each node selects a MAP resulting from a global optimization.
This is not an easy task in general as
optimization schemes aiming at maximizing some global utility lead to
intricate correlations between the behaviors of all nodes and these
correlations most often prevent the use of the independence based tools
underlying most of random graph theory and stochastic geometry.

The present paper identifies a class of local adaptation policies for which  
the optimal node MAP control leads to closed form expressions for both
the behavior of a typical node and global performance
metrics; these closed forms are obtained under the 
assumption of an {\em infinite random network} whose nodes have locations 
forming a realization of a homogeneous PPP in the Euclidean plane. 
These adaptive policies are defined in term of {\em stopping sets} (see below).
A typical example of stopping set is the smallest disk that contains the
$k$ closest receivers/neighbors of a node. The associated adaptive policy 
consists in having each transmitter controlling his MAP
as a function of the {\em actual} locations of the $k$ closest receivers
which are the closest to its own location and the {\em estimated density}
of the other receivers, and in adopting the 
MAP which maximizes the sum of the logarithms of
the rates of all nodes in view of this information and estimation.
Another example of stopping set is any deterministic set (e.g. a disk). 
The class is however much larger than what is suggested by these
two simple examples, as we shall see below.

So far only two extreme cases are understood within this general
adaptive Aloha framework: (1) the case with empty stopping sets: this 
form of open loop control can be analyzed using the tools
of stochastic geometry and parametric optimization alluded to above;
(2) the case with $\R^2$ stopping sets:
this form of adaptive control, which is based on full information on the network topology,
was quite recently analyzed in \cite{BacSing}.
As in many optimal control problems, none of these extreme cases are
really satisfactory. The former is suboptimal: it is well known that the knowledge
of actual location information should allow one to quite significantly improve overall network
performance; the latter is not implementable as it requires
location informations of all the nodes in the network (each
transmitter has to acquire knowledge
on either the locations of all receivers or its path-loss to all
receivers, which form infinite sequences in both cases).

The main new mathematical results obtained in the present paper
are (i) the characterization of the optimal policy for any (translation invariant) 
stopping set information structure [Theorem \ref{t:opt-pi}] and
(ii) closed form expressions for the resulting
performance metrics [Theorems \ref{thm41}--\ref{thm44} below]. These, in the first place, constitute an important
theoretical extension of what is currently known on Spatial Aloha.
In addition, as we shall see, the locality of the stopping set can be
used to define implementable control schemes, and
our analytical tools can also be used to find a good trade-off
between the quality of the optimal strategy and the locality
of the stopping sets.
 
Let us now give a somewhat broader review of the state of the art
on the performance of Aloha, Spatial Aloha and related protocols. Aloha and slotted Aloha
were introduced and analyzed by Abramson~\cite{ctrltheory-comnets.abramson70aloha}
and Roberts~\cite{ctrltheory-comsnets.roberts75aloha-slots-capture}
respectively. In these protocols, only one node could successfully
transmit at a time. In Spatial Aloha, as
considered in~\cite{ctrltheory-wireless.baccelli-etal06aloha-multihop-wireless},
nodes are located in the Euclidean space, signal power is attenuated according to 
some path-loss function and the success of each transmission is based on
Signal to Interference and Noise Ratio (SINR).
This leads to an essential property of wireless networks
which is spatial reuse, namely the fact that infinitely many simultaneous successful
transmissions can take place in the whole space, provided they are ``well separated''.
More precisely, it was shown that such networks can sustain a positive
spatial density of successful transmissions~\cite{ctrltheory-wireless.baccelli-etal06aloha-multihop-wireless}.
All the above protocols prescribe
identical MAPs for all nodes.

Among the initial attempts of MAP adaptation in Aloha,
let us first quote~\cite{ctrltheory-wireless.hazek85-decentralized-control-multiaccess}
which analyzes stochastic approximation based strategies
adapting the MAP to receiver feedbacks and aiming at stabilizing the network.
This approach was considered in the spatial random network context
in \cite{BordenaveFS07} for the protocol model but not for the SINR setting.
The other recent publications on the matter can be organized as follows: 
(1) {\em Adaptation to local channel conditions}. In spatial random networks with SINR based reception,
such local adaptations are amenable to parametric optimization as described above~\cite{ji-we-and, ctrltheory-wireless.baccelli-etal09spatial-opportunistic-aloha}. Some versions of this type of adaptation are referred
to as Oppotunistic Aloha \cite{stochproc-wireless.baccelli-blaszczyszyn09stochastic-geometry-wireless-networks-2}.
The case with no dependence on geometry~(i.e., no path-loss component)
is also considered in~\cite{ctrltheory-wireless.hsu-su11channel-aware-aloha} 
where a paradoxical behavior is identified for certain topologies:
plain Aloha yields better aggregate throughput than opportunistic Aloha.
(2) {\em Centralized optimization for fixed topologies}. 
The case with SINR based reception is studied
in~\cite{ctrltheory-wireless.mohsenian-rad-et-al10sinr-random-access}, where
the authors give centralized algorithms 
determining the random access probabilities 
that maximize either the total network throughput
or lead to a max-min fair operation. By nature, this approach is difficult to use
in large random spatial networks.
(3) {\em Interference graph optimization}. There is a vast literature on this topic
which replaces SINR reception by exclusion rules on a graph.
Among the most relevant papers for our purposes, let us
quote~\cite{ctrltheory-wireless.wang-kar04max-min-fair-aloha}
which proposes algorithms that lead to either proportional fairness or max-min fairness
and~\cite{gupta-stolyar12capacity-region-random-access} which 
derives the Pareto boundary of the rates achievable by Aloha.
(4) {\em Game theoretic analysis}. In this appraoch, the MAP adaptation 
is modeled by a game. Again, we will limit ourselves to
papers investigating the case of spatial random networks. The most relevant
references are~\cite{gametheory-wireless.hanawal-etal12medium-access-games}
and \cite{zhang-haenggi12power-control-poisson-networks}.
Reference~\cite{gametheory-wireless.hanawal-etal12medium-access-games}
formulates the channel access problem as a non-cooperative game among users and
proposes pricing schemes that induce a socially optimum behavior at equilibrium.
However the approach is restricted to
symmetric Nash equilibria, which forbids the type of non-symmetrical
adaptations of interest here. 
Reference~\cite{zhang-haenggi12power-control-poisson-networks} shows
that Aloha is often a Pareto optimal strategy when users aim at selecting 
their transmission powers in a game theoretic way.
A key difference between game theory and the line of thought
of the present paper is that users are selfish in the former and altruistic
in the latter (as explained above, as in TCP, they willingly throttle their
rate to maximize some utility).

Let us summarize this survey by stressing that,
excluding \cite{BacSing} discussed above, 
none of the above references considers a large random Aloha network where nodes adapt to wireless channel
randomness and to spatial fluctuations of topology for making
non asymmetric random access decisions maximizing
some utility, as we do in the present paper.

There is also a vast literature on the modeling
of CSMA in large random networks by stochastic geometry.
In spite of the fact that
the very nature of this MAC protocol is adaptive,
we will not discuss this here since
CSMA is only designed to guarantee a reasonable scheduling,
not to optimize any utility.

The present paper is structured as follows:
the network setting is described in Section~\ref{sec:network-model}.
There, we describe the quasi-static networks of interest, in which mobiles
learn the local topology and incorporate this information
in their MAP selection.
Section~\ref{sec:prop-fair-aloha} is focused on the stopping set based distributed
algorithms that lead to a proportional fair sharing of the network resources.
We show that nodes can compute the optimal MAPs as solutions to certain
fixed point equations.
Section~\ref{ss:stocg} contains the analytical performance results.
For nodes forming a realization of a
homogeneous Poisson Point Process (PPP) in the Euclidean plane,
we compute the MAP distribution using the theory of shot noise fields.
Using this distribution, we also derive
the mean utility of a typical node.
The numerical results are gathered in Section~\ref{ss:numerics}.
The aim of this section is three-fold: (1) validate the analytical
results against simulation; (2) quantify the gains brought
by adaptation;
(3) determine a good trade-off between 
information and performance. 
Finally, we discuss the implementability of this class of 
controls in Section \ref{sec:impl}.
\section{Network Model}
\label{sec:network-model}


We model the ad-hoc wireless network as a set of transmitters
and their corresponding receivers, all located in the
Euclidean plane. This is often referred to as the ``bipole model''~\cite[Chapter~16]{stochproc-wireless.baccelli-blaszczyszyn09stochastic-geometry-wireless-networks-2}.
We assume that each node has an infinite backlog of packets to transmit to its
receiver. The transmitters follow the
slotted version of the Aloha medium access control~(MAC) protocol.
A transmitter, in each transmission
attempt, sends one packet which occupies one slot. The transmission succeeds if the
signal to noise plus interference ratio~(SINR) at the corresponding receiver exceeds a designated threshold.


We assume that the transmitting nodes are scattered according to a homogeneous Poisson
point process of intensity $\lambda$. Each transmitter is associated
with a multidimensional mark that carries information
about the location of the corresponding receiver, the fading
conditions of the channels to all the receivers and about its MAC
status at the current time slot. More precisely, we have at our disposal
a marked Poisson point process $\tilde{\Phi} = \{X_i,y_i,{\bf
  F}_i,e_i,p_i\}$. Below,
\begin{itemize}
\item $\Phi = \{X_i\}$ denotes the PPP of
intensity $\lambda$ representing the location of transmitters in $\R^2$.
\item $\bar{\Phi} = \{y_i\}$ denotes the locations of the receivers. The receiver of transmitter $X_i$ 
is $y_i$. For notational
convenience, we also define random vectors $\{r_i = y_i - X_i\}$. We assume that
$|r_i| = r$, a constant, and that $\{\angle r_i\}$
are i.i.d. and uniform on $[0,2\pi]$ and independent of $\Phi$. 
From the displacement theorem~\cite{stochproc-wireless.baccelli-blaszczyszyn09stochastic-geometry-wireless-networks-1}, $\bar{\Phi}$
is also a homogeneous PPP of intensity $\lambda$.
\item $\{{\bf F}_i = (F_i^j:j)\}$ where $F_i^j$ denotes the random
  fading between transmitter $j$ and receiver $i$. We assume that channels are Rayleigh faded and that the random variables
$(F_i^j:i,j)$ are independent and exponentially distributed with mean $1/\mu$.
The independence of the components of $(F_i^j:i,j)$ is justified in moderately dense networks.
\item $\{e_i\}$ are indicators that take value
one if the corresponding node decides to transmit in the considered
time slot, and zero otherwise.
\item $p_i$ denotes the MAP of node $i$:
$p_i = \mathbb{P}(e_i = 1) = 1 - \mathbb{P}(e_i = 0)$. Also,
given $p_i$, $e_i$ is independent of everything else
including $\{e_j\}_{j \ne i}$.
\end{itemize}
We assume that each transmitter uses unit transmission power.
We also adopt an omnidirectional path-loss model with power attenuation
given by $l(z) = z^{-\beta}$ for a $\beta > 2$. The receivers are
also subjected to white Gaussian thermal noise with variance
$W$, which is also constant across slots. Thus, when a transmitter $i$
transmits, the SINR at its receiver is
\[
\sinr_i = \frac{F_i^i r^{-\beta}}{\sum_{j \neq i}|X_j - y_i|^{-\beta}F_j^i e_j + W},
\]
where the first term in the denominator is the power of the interference at receiver $i$, which is
a shot noise associated with $\tilde{\Phi}$ (see e.g. \cite{stochproc-wireless.baccelli-blaszczyszyn09stochastic-geometry-wireless-networks-1}).
We assume that a transmission from $i$ to its receiver
is successful if $\sinr_i$ exceeds some threshold $T$.
\section{Proportionally Fair Aloha}
\label{sec:prop-fair-aloha}

The principle of proportional fairness~\cite{Kelly}
consists in maximizing the sum of the
logarithms of user  throughputs.
\subsection{Node Throughput}
In our model, the throughput of
transmitter $i$ (actually, transmitter-receiver pair $i$)
is $p_i q_i$, where $q_i$ is the probability of
successful transmission of node $i$  given this node is
authorized to transmit. More specifically, we consider the conditional
probability of successful transmission of node $i$  {\em given the network
geometry} $\Phi,\bar\Phi$:  $q_i=q_i(\Phi,\bar\Phi):=\Pr\{\mathrm{SINR}_i\ge
T\,|\,\Phi,\bar \Phi\}$.
Let $b_{ji} := |X_j - y_i|^{\beta}/(T r^{\beta})$.
The conditional probabilities $q_i$ admit the following expression,
valid for {\em arbitrary} (fixed) network topology
$\Phi,\bar\Phi$, with either a finite or an infinite number of nodes:
\begin{lemma}
\begin{equation}
\label{eqn:succ-prob}
q_i = e^{-\mu T r^{\beta}W}
\prod_{j \neq i}\left(1 - \frac{p_j}{1 + b_{ji}}\right).
\end{equation}
\end{lemma}
\begin{IEEEproof}
Condition first  on $\Phi, \bar\Phi$ and  $\calF_i=\{F_j^i,e_j, j \neq i\}$:
\begin{align}
\lefteqn{\mathbb{P}\left\{\sinr_i \geq T \,|\, \Phi,\bar\Phi,\calF_i\right\}} \nonumber \\
      & = \mathbb{P}\left\{F_i^i \geq T r^{\beta}\left(\sum_{n \neq i}|X_j - y_i|^{-\beta}F_j^i e_j + W\right) \Big| \Phi,\bar\Phi,\calF_i\right\}  \nonumber \\
      & = e^{-\mu T r^{\beta}W}
e^{-\sum_{j \neq i} \frac{\mu T r^{\beta}F_j^i e_j}{|X_j - y_i|^{\beta}}},  \nonumber
\end{align}
where the last expression follows from the fact that $F_i^i$ is
exponential, independent of $\calF_i$. Averaging over
$\calF_i$, we obtain
$$ \prod_{j \neq i}
\mathbb{E}\left[e^{-\frac{\mu T r^{\beta}F_j^i e_j}{|X_j - y_i|^{\beta}}}\right]
= \prod_{j \neq i} \left(1 - p_j + \frac{p_j}{1 + 1 / b_{ji}} \right)\,. 
$$
We obtain~(\ref{eqn:succ-prob}) by further simplifying the product factors.
\end{IEEEproof}

Thus the proportional fairness of our model consists
in finding the $p_i$s, $0\le p_i\le 1$,  which maximize the sum of the
logarithms of the throughputs   $\sum_i \log(p_iq_i)$.
However, in  infinite networks,  this sum
is typically unbounded. In particular, it is unbounded for almost all
realizations of our   Poisson network of
Section~\ref{sec:network-model}.
Below, we first propose a specific
formulation of the proportional fairness in the context of an infinite
ergodic model. Subsequently, we show that the solution of this
problem asymptotically coincides  with a  meaningful
maximization of the sum of the logarithms of the throughputs in a
finite network when its size tends to infinity.
\subsection{MAC Policies}
Denote  by $\calS_{a}$ the {\em shift} operator on  $\Phi,\bar\Phi$. It
translates all the network nodes by the vector $-a$:
$\calS_a\{X_i\}=\{X_i-a\}$ and similarly for the receivers, preserving
all the nodes characteristics. We extend this operator to all (random
or deterministic) subsets $A\subset\R^2$ by defining $\calS_aA=\{x-a: x\in A\}$.
An important assumption  in the infinite model is that
all nodes  behave in the same way if they thy see the same
configuration of nodes in $\Phi,\bar\Phi$. 
More precisely, by {\em Translation invariant MAC policy},
we mean a policy where all nodes $X_i$ set their MAPs to
$p_i=\psi(\calS_{X_i}\Phi,\calS_{X_i}\bar\Phi)$, where  $\psi(\cdot)$
is some given function which takes as its argument the network geometry
$(\Phi,\bar\Phi)$ and has its values in $[0,1]$.
In other words, any node $X_i$, in order to chose its MAP,
applies the same {\em MAC policy} $\psi$ evaluated for the  network geometry
``seen'' from its point of view, i.e. from $X_i$.

Usually a given node will only have some partial information
about the location of other nodes, e.g. limited to some geometric vicinity.
The notion of ``local spatial information'' can be formalized
using the notion of {\em stopping set} $S=S(\Phi,\bar\Phi)$
(cf. \cite[Definition~1.9]{stochproc-wireless.baccelli-blaszczyszyn09stochastic-geometry-wireless-networks-1}).
This is a subset of the plane $\R^2$ such that  for any observation window  $A$
one can determine whether $S(\Phi, \bar\Phi)\subset A$ when knowing
only the points of $\Phi,\bar\Phi$ in $A$.%
~\footnote{The notion of
  stopping set (with respect to a spatial  point pattern)
is a spatial analogue of the stopping time for a temporal process.}
The stopping set  $S$ models the region in which the locations of nodes
are known to a (hypothetical) observer located at the origin.
In this context, it is natural (but not necessary) to assume that $S$ is some
neighborhood of the origin.
The simplest examples of $S(\Phi,\bar\Phi)$ are disks
centered at the origin of fixed radius $R$ or of  radius equal to
the distance to the $n-$th closest transmitter, or $n-$th  closest receiver
or $n-$th closest node (regardless whether it is a transmitter or a receiver).

For a given stopping set $S=S(\Phi,\bar\Phi)$ we consider the following class
of {\em MAC policies with local spatial information $S$}. By this we mean
a (translation invariant) MAC policy
$\psi$ satisfying the following constraint
\begin{equation}\label{e.pi-S}
\psi(\Phi,\bar\Phi)=\psi(\Phi\cap S, \bar\Phi\cap S).
\end{equation}
Note that in order to apply such a policy and evaluate its MAP
$p_i=\psi(\calS_{X_i}\Phi,\calS_{X_i}\bar\Phi)=
\psi(\calS_{X_i}\Phi\cap S_i, \calS_{X_i}\bar\Phi\cap S_i)$,
node $X_i$ needs to know only the locations of
the other nodes in the stopping set $S_i:=S(\calS_{X_i}\Phi,\calS_{X_i}\bar\Phi)$
(e.g. in a fixed disk around it,
the disk up to its  $n-$th closest neighbor, etc.)
\subsection{Infinite Network Optimization under Spatial Constraints}\label{ss.ergodic}
Denote by  $\mathbb{P}^0$ the
Palm distribution  of the stationary marked point
process $\tilde\Phi$. Recall that
almost surely  under $\mathbb{P}^0$ there is a node located at the origin $X_0=0$, called the {\em typical node}. By  Slivnyak's
theorem, in the case of the Poisson process, this is just an ``extra''
node added at the origin to the stationary configuration of nodes,
with all its characteristics independent and distributed identically
as for all other nodes. Denote by $\E^0$ the expectation with respect
to $\Pr^0$.

Let a  stopping set $S=S(\Phi,\bar\Phi)$ be given.
Denote by ${\bf PF}^S$ the following
maximization problem in the MAC policy $\psi(\cdot)$
\begin{align}
{\rm {\bf PF}^S:} \quad \maximize& \mathbb{E}^0[\log(p_0 q_0)]  \label{eqn:prop-fair-objective-ergodic} \\
 \subjectto & 0 \leq p_i=\psi(S_{X_i}\Phi,S_{X_i}\bar\Phi) \leq
 1\nonumber \\
\operatornamewithlimits{and~~~}& \psi(\cdot)
\operatornamewithlimits{satisfying~}~(\ref{e.pi-S}).
  \nonumber
\end{align}
We call ${\bf PF}^S$ {\em the proportional fair Aloha problem
with spatial information $S$}.

Note that under $\mathbb{P}^0$ we have $p_0=\psi=\psi(\Phi,\bar\Phi)$
and that~(\ref{eqn:prop-fair-objective-ergodic}) corresponds to
the {\em maximization of the expected logarithm of the throughput of
  the typical node}. The fact that the maximization is done  with respect to
a policy $\psi(\cdot)$ applied by all the nodes in the network makes
it non-trivial, despite the fact that we maximize the utility
function for just one node $X_0=0$. 
Note also that $\mathbb{E}^0[\log(p_0 q_0)]\le 0$.

For a given  stopping set $S=S(\Phi,\bar\Phi)$ consider  the
following MAC policy. Define $\psi^S$  as  the (unique) solution of
 \begin{equation}
 \label{eqn:opt-p0}
 \frac{1}{\psi}  = \sum_{y_j\in S, j\not=0}\frac{1}{1 + b_{0j} - \psi}
+\int\limits_{y \in \mathbb{R}^2\setminus S}\frac{\lambda{\rm d}y}{1+|y|^\beta/Tr^{\beta} - \psi}
 \end{equation}
in $\psi$ provided
\begin{equation}
 \label{eqn:bountfp}
 a_i:=
\sum_{y_j\in S, j\not=0}\frac{1}{b_{0j}}+\int_{y \in \mathbb{R}^2\setminus S}\frac{\lambda{\rm d}y}{|y|^\beta/Tr^{\beta}} >1
 \end{equation}
and  $\psi^S = 1$  otherwise.
The existence and uniqueness follow
from the fact that the L.H.S. of (\ref{eqn:opt-p0})
decreases from $\infty$ to 1 w.r.t. $\psi$ on $[0,1]$
whereas the right hand side (R.H.S.) increases to
$a_i$ on $[0,1]$, and from the continuity of these two functions.

Note that $\psi^S$ satisfies~(\ref{e.pi-S}). In fact
$\psi^S(\Phi,\bar\Phi)=\psi(\bar\Phi\cap S)$ depends only on the configuration of
receivers in $S$.

We now state now the main structural result of this paper.
Its proof is given in Appendix.
\begin{theorem}\label{t:opt-pi}
For all given stopping sets $S=S(\Phi,\bar\Phi)$, the MAC policy
$\psi^S$ defined by the fixed point equation (\ref{eqn:opt-p0})
is a solution of the  proportional fair Aloha problem ${\bf PF}^S$
with spatial information $S$. For this MAC policy $-\E^0[\psi^S q_0]<\infty$.
Moreover, for any  MAC policy $\psi'(\cdot)$ solving  ${\bf PF}^S$
we have  $\psi'(\Phi,\bar\Phi)=\psi^S(\Phi,\bar\Phi)$ for
almost all realizations of $(\Phi,\bar\Phi)$.
\end{theorem}

\subsection{Extended Window Approach}
The proportional fair medium
access problem is usually stated as follows in the context of a {\em finite
subset of network nodes}. Consider some selected (arbitrarily chosen) nodes
$X_1,\ldots,X_N$ from $\Phi$, with their MAPs $p_i$ and success
probabilities $q_i$ given by~(\ref{eqn:succ-prob}). Consider  the
optimization problem
\begin{align}
{\rm {\bf PF}^{\ast}:} \quad \maximize & \sum_{i\in[N]} \log{(p_i q_
  i)}, \label{eqn:prop-fair-objective}\ \inargument \{p_i\}\\
 \subjectto & 0 \leq p_i \leq 1, \forall i\in[N], \nonumber
\end{align}
where $[N]=\{1,\ldots,N\}$.
Note that here, we do not restrict ourselves to 
translation invariant MAC policies. 
${\rm {\bf PF}^{\ast}}$
is a convex separable optimization problem. The optimal MAPs can be
characterized as follows, cf.~\cite{BacSing}:
the unique solution of $\rm {\bf PF}^{\ast}$
is the unique solution $p_i$, $1 \leq i \leq N$, of
 \begin{equation}
\label{eqn:opt-p-finite}
 \frac{1}{p_i} = \sum_{j  \in [N]\setminus \{i\}}\frac{1}{1 + b_{ij} - p_i},
 \end{equation}
provided  $a_i= \sum_{j \in [N]\setminus \{i\}}1/b_{ij} > 1$
and  $p_i = 1$  otherwise.

The thermal noise $W$ appears
merely in a constant additive term in the
objective function. So, this does not affect the optimal MAPs. For the
same reason interferers $X_j$, $j\not\in[N]$, external to the selected subset
of nodes $X_1,\ldots,X_N$, do  not affect the optimal MAPs of these
nodes. As in the previous section,
the optimal MAP of a transmitter $i$ is only a function of other
receivers' locations $\{y_j\}_{j \in [N]\setminus \{i\}}$. In
particular, given $\{y_j\}_{j \in [N]\setminus \{i\}}$, $p_i$
does not depend on $\{X_j\}_{j \in [N]\setminus \{i\}}$.

Now, consider a bounded observation window $A$ on the plane and the solution
$\{p^A_i:X_i\in A\}$ of the ${\rm {\bf PF}^{\ast}}$ problem posed for
the (finite set of) nodes $X_i\in A$. Call it ${\rm {\bf PF}^{A}}$
problem. The following corollary follows immediately from of the
above characterization of the solution of~${\rm {\bf PF}^{\ast}}$:
\begin{corollary}
\label{cor:eqaciter}
For any given $X_i\in\Phi$,  when $A$ increases to $\R^2$,
$p_i^A$ converges to $p_i^{\R^2}$ which is the unique solution (in $p_i$) of
(\ref{eqn:opt-p-finite}) with the summation carried out over
all $j \neq i$.
\end{corollary}
\begin{remark}
Note that $p_i^{\R^2}=\psi^{\R^2}(\calS_{X_i}\Phi,\calS_{X_i}\bar\Phi)$, i.e.,
the  solution of the  proportional fair Aloha problem
with complete  geometry information ($S=\R^2$)
considered in Section~\ref{ss.ergodic}
prescribes MAPs that are equal to the limits of the solutions
of the finite problem  ${\rm {\bf PF}^{A}}$ when $A\to \R^2$.
\end{remark}

Finally, the maximization of the mean utility of
the typical node~(\ref{eqn:prop-fair-objective-ergodic})  can be also
interpreted (via ergodicity) as the spatial average of node utilities
in the extended window. Indeed, note that the maximization of
$\sum_{X_i\in A} \log(p_iq_i)$ for a given bounded window  $A$ is equivalent
to that of $\Theta_A=1/(\lambda |A|)\sum_{X_i\in A} \log(p_iq_i)$, where
$|A|$ is the surface of $A$ and $\lambda$ is the transmitter  density.
When $A\to \R^2$, $\Theta_A$ converges to
$\mathbb{E}^0[\log(p_0 q_0)]$ in the case when
$\psi_i=\psi(\calS_{X_i}\Phi,\calS_{X_i}\bar\Phi)$.
\subsection{Examples of Stopping Sets}
\label{sec:exa1}
We illustrate this via a few examples. We limit ourselves to
disk based sets, which makes sense in the isotropic
and omnidirectional setting considered here.
But the proposed framework is quite versatile and
could accommodate more general situations.
We start with deterministic stopping sets and then consider random ones.
In both cases, we illustrate the continuum alluded to above
by increasing levels of information.

Below, we will denote by $R_p$ the distance between $X_0$, the tagged
transmitter, and its $p$-th closest receiver, excluding $y_0$, with
$p$ a positive integer. Let $F(\psi ,S)$ denote the R.H.S.
of~\eqref{eqn:opt-p0} for a general stopping set $S$. 
With this notation, the tagged transmitter's MAP satisfies the fixed point equation
\begin{equation}
\label{eq:bfpe}
\frac 1 {\psi} = F(\psi ,S).
\end{equation}
We will use the following notation:
\begin{eqnarray*}
D(\psi,S) & = &
\sum_{j \neq 0: y_j \in S} \frac{1}{\frac{(|y_j|/r)^{\beta}}{  T}  + 1-\psi },\\
C(\psi,x) & = & 
2 \pi \lambda r^2 \int_x^{\infty}\frac{s}{\frac{s^{\beta}}{  T} + 1 - \psi } {\rm d}s.
\end{eqnarray*}

\subsubsection{$S = \emptyset$}
This is the case where transmitters have no topological
information at all. In this case, all transmitters use an identical MAP
given by (\ref{eq:bfpe}) with
$F(\psi ,S) = C(\psi ,0).$
Notice that in this case, $\psi  < 1$
irrespective of the node density $\lambda$.
Further, this MAP is different from the one that maximizes the density of successful
transmissions~(see~\cite[Section~16.3.1.1]{stochproc-wireless.baccelli-blaszczyszyn09stochastic-geometry-wireless-networks-2}), which makes sense as the objective functions of the parametric optimization are different.
For the special case $\beta = 4$, we get:
\begin{equation*}
 \psi^{\emptyset}  = \frac{\sqrt{1+4\alpha^2} -1}{2 \alpha^2}, \quad
\mbox{where} \quad
\alpha = \frac{\pi^2 \lambda r^2 \sqrt{T}}{2}.
\end{equation*}
\subsubsection{$S = B_0(R)$}
In this example, each transmitter knows the locations of
all receivers in a disk of radius $R$ centered on its location,
with $R$ a constant parameter in $(0,\infty)$.
This parameter can be tuned to control the mean cardinality of
the set of receivers taken into account in the adaptive control.
An inconvenience of this setting is that this cardinality is random
(it follows a Poisson law). The parameter $R$ can also
be used in order to upper-bound by $R^\beta$
the mean path-loss between the tagged
transmitter and this set of receivers.
This will be important in the proposed implementation
(see Section \ref{sec:impl}).
In this case $\psi$ is solution of (\ref{eq:bfpe}) with
\begin{align}
\label{eq:b0R}
F(\psi ,S) = D(\psi ,B_0(R))+ C(\psi ,\frac R r).
\end{align}
For $\beta = 4$, we can further simplify the second term as
\begin{equation} \hspace{-.2cm}C(\psi ,\frac R r) =  
\label{eq:beta4}
\frac{\pi \lambda r^2 \sqrt{T}}{\sqrt{1-\psi }} \left(\frac{\pi}{2} - \tan^{-1}\left(\frac{(R/r)^2}{\sqrt{T(1-\psi )}}\right)\right).
\end{equation}

\subsubsection{$S = \mathbb{R}^2$}
Here, each transmitter knows and uses the
location information of all receivers in the network.
In this case, $F(\psi ,S)$ is equal to the R.H.S.
of~\eqref{eqn:opt-p-finite} with the summation carried out over
all $j \neq i$. This is the case studied in \cite{BacSing}.

\subsubsection{$S = B_0(R_1)$}
This is our first random set example: the tagged transmitter only
knows/uses the location of it {\em nearest neighboring receiver},
excluding $y_0$. We then have (\ref{eq:bfpe}) with
\begin{equation}
\label{eq:firstexa}
 F(\psi ,S) = \frac{1}{\frac{(R_1/r)^{\beta}}{  T} + 1-\psi } +
C(\psi ,\frac {R_1} r).
\end{equation}
It is easy to see that $\psi (R_1)$ increases with $R_1$.
\subsubsection{$S = B_0(R_k)$}
Here $k$ is a positive integer. The case $k=0$ boils down to $S=\emptyset$
and the case $k=\infty$ to $S=\R^2$. In terms of implementation,
this setting is better than $S=B_0(R)$ 
since it allows one to tune the actual cardinality of the set of receivers
taken into account in the MAP control.
However there
is no guarantee on the path-loss between the transmitter and the nodes in this set.
In this case, we have (\ref{eq:bfpe}) with
\begin{equation}
 F(\psi ,S) = \sum_{p=1}^k
\frac{1}{\frac{R_p^{\beta}}{r^\beta  T} + 1-\psi } +
C(\psi,\frac{R_k}r).
\end{equation}
\subsubsection{$S = B_0(R_k)\cap B_0(R)$}
This two parameter stopping set enjoys the two practical properties mentioned above:
(i) the number of receivers taken into account in the control of each transmitter
is upper bounded by $k$; (ii) the mean path-loss to any receiver taken into account in the control
is upper bounded by $R^\beta$.
In this case, we have (\ref{eq:bfpe}) with
\begin{eqnarray}
\label{eq:lastexa}
 F(\psi ,S)  = \hspace{-.7cm} \sum_{p=1}^{\min(k,B_0(R))}
\frac{1}{\frac{R_p^{\beta}}{r^\beta T} + 1-\psi }
+C(\psi,\frac{\min(R_k,R)}r).
\end{eqnarray}
If $\beta=4$ 
the R.H.S.s of (\ref{eq:firstexa})-(\ref{eq:lastexa}) can be simplified 
using (\ref{eq:beta4}).

\section{Performance of the Optimal Control}
\label{ss:stocg}
\subsection{MAP Distribution}
\label{ss:mapdis}
Using Theorem~\ref{t:opt-pi} and monotonicity arguments,
it is not difficult to show that for all $0 \leq \rho \leq 1$
\begin{equation}
\label{eq:key}
\psi  > \rho \quad \mbox{iff} \quad
\sum_{\overset{j \neq 0:}{y_j \in S}}\frac{\rho}{b_{0j} + 1- \rho} + I(\rho,S) < 1,
\end{equation}
where
$I(\rho,S) = \lambda\int_{y \in \mathbb{R}^2\setminus S}\frac{\rho}{|y|^\beta/Tr^{\beta} + 1 - \rho}{\rm d}y.$
We now use this to derive the distribution of the optimal MAP.
For all $0 \leq \rho \leq 1$ and stopping sets $S$, let
\[L_{(\rho,S)}(x,y) = \frac{\rho}{\frac{| x- y|^{\beta}}{T r^{\beta}} + 1 - \rho} \mathbf{1}(y \in S).\]
The indicator $ \mathbf{1}(y_i \in S)$ can be thought as a mark associated
with $y_i \in \bar{\Phi}$. However, these marks are not independent
unless $S$ is a constant set, e.g., if $S = B_0(R_1)$.
For all $0 \leq \rho \leq 1$ and stopping sets $S$, the shot noise field
$J_{\bar{\Phi}}(\rho,S)$ associated with the
response function $L_{(\rho,S)}(0,y)$ and the
marked point process $\bar{\Phi}$ is
$J_{\bar{\Phi}}(\rho,S) = \sum_{i} L_{(\rho,S)}(0,y_i)$.
Notice that this shot noise is not that representing the interference
at the origin. It rather measures the effect of the presence of a transmitter
at $0$ on the set of receivers in $S$. We have the following connection
between the optimal MAP distribution and the shot noise
$J_{\bar{\Phi}}$. For all $0 < \rho < 1$,
$\mathbb{P}^0(\psi  > \rho) = \mathbb{P}^0\left(J_{\bar{\Phi} \setminus {\{y_0\}}}(\rho,S) < 1 - I(\rho,S)\right),$
and
$\mathbb{P}^0(\psi  = 1) = \mathbb{P}^0\left(J_{\bar{\Phi} \setminus {\{y_0\}}}(1,S) < 1 - I(1,S)\right)$.
From Slivnyak's theorem~\cite[Theorem~1.13]{stochproc-wireless.baccelli-blaszczyszyn09stochastic-geometry-wireless-networks-1},
$\mathbb{P}^0\left(J_{\bar{\Phi} \setminus {\{y_0\}}}(\rho,S) \! < \! 1 - I(\rho,S)\right) = \mathbb{P}\left(J_{\bar{\Phi}}(\rho,S) \! < \! 1 - I(\rho,S)\right)$,
for all $0 \leq \rho \leq 1$. Thus, we get:
\begin{theorem}
\label{thm41}
For all $0 < \rho < 1$,
$\mathbb{P}^0(\psi  > \rho) = \mathbb{P}\left(J_{\bar{\Phi}}(\rho,S) < 1- I(\rho,S)\right)$ and
$\mathbb{P}^0(\psi  = 1) = \mathbb{P}\left(J_{\bar{\Phi}}(1,S) < 1- I(1,S)\right).$
\end{theorem}
\subsection{Examples}
\subsubsection{Deterministic Sets}
Since $\bar{\Phi}$ is a homogeneous Poisson point process, it follows
from~\cite[Proposition~2.6]{stochproc-wireless.baccelli-blaszczyszyn09stochastic-geometry-wireless-networks-1} that
the Laplace transform $\mathcal{L}_{J(\rho,S)}(s)$
of the shot noise $J_{\bar{\Phi}}(\rho,S)$ is:
\begin{equation}
\mathcal{L}_{J(\rho,S)}(s) = \exp \left(-\lambda \int_S\left(1 - \e^{-\frac{s\rho T r^{\beta}}{| y|^{\beta} + (1 - \rho)T r^{\beta}}}\right){\rm d} y \right)\!\!. \!\!
\label{eq:lapl}
\end{equation}
For example, when $S = B_0(R)$ for a fixed $R$,
\begin{equation*}
\mathcal{L}_{J(\rho,S)}(s) = \exp \left(-\pi\lambda \int_0^{R^2}\left(1 - \e^{-\frac{s\rho T r^{\beta}}{t^{\beta/2} + (1 - \rho)T r^{\beta}}}\right){\rm d} t \right).
\end{equation*}
Moreover, if $\beta = 4$, the Laplace transform
$\mathcal{L}_{J(\rho,S)}(s)$ can be further simplified as
$$\mathcal{L}_{J(\rho,S)}(s)  
 = e^{-2 \pi \lambda \sqrt{(1-\rho)T}r^2 \int_{v_R}^1 \frac{1 - \e^{-s\rho v^2/(1-\rho)}}{v^2 \sqrt{1 - v^2}} {\rm d} v },
$$
where
$
v_R =  {\sqrt{(1-\rho)T}r^2}/{\sqrt{R^4 + (1-\rho)Tr^4}}.
$
\begin{theorem}
\label{thm:map-distribution}
For all deterministic sets $S$,
the optimal MAP of the typical node has for distribution
\begin{equation}
\!\!\!\mathbb{P}^0(\psi  > \rho) =\frac 1 {2\pi}
\int_{-\infty}^{\infty}\mathcal{L}_{J(\rho,S)}(i w)
\frac{\e^{iw(1 - I(\rho,S))} -1}{iw} {\rm d} w,\!\!\!\!\!\!
\label{eq:Parsev}
\end{equation}
with $\mathcal{L}_{J(\rho,S)}(\cdot)$ given by~\eqref{eq:lapl}. Similarly, $\mathbb{P}^0\{\psi  = 1\}$
is given by the R.H.S. with $\rho$ replaced by $1$.
\end{theorem}
\begin{proof}
Let $g_{\rho}(\cdot)$ denote the density of 
$J_{\Phi}(\rho,S)$. Then
\begin{equation*}
\mathbb{P}^0(\psi  > \rho) = \int_0^{1 - I(\rho,S)}g_{\rho}(t) {\rm d} t
                         = \int_{-\infty}^{\infty} g_{\rho}(t) u(t)  {\rm d} t,
\end{equation*}
where $u(t) = 1$  if $ 0 \leq t \leq 1 - I(\rho,S)$ and 0 otherwise.
Now from the Plancherel-Parseval
theorem~(see~\cite[Lemma~12.1]{stochproc-wireless.baccelli-blaszczyszyn09stochastic-geometry-wireless-networks-1})
\[
\mathbb{P}^0(\psi  > \rho) = \frac{1}{2\pi}\int_{-\infty}^{\infty} \mathcal{F}_{J(\rho,S)}(w) \mathcal{F}^*_{u}(w) {\rm d} w,
\]
with $\mathcal{F}_{A}(w)=\mathbb{E}\exp(-iwA)$ the Fourier
transform of the real valued random
variable $A$ and $B^*$ the complex conjugate of $B$.
The claim follows after substituting
$ \mathcal{F}_{u}(w) = \frac{1 - \e^{-iw(1 - I(\rho,S))}} {iw}$
and $\mathcal{F}_{J(\rho,S)}(w) = \mathcal{L}_{J(\rho,S)}(iw)$.
Similarly, $\mathbb{P}^0(\psi  =  1)$.
\end{proof}
\subsubsection{$S = B_0(R_1)$}
Observe that
$
I(\rho,B_0(x)) = 2 \pi \lambda r^2 \int_{x/r}^{\infty}\frac{\rho s}{\frac{s^{\beta}}{  T} + 1 - \rho} {\rm d}s
$
is increasing in $\rho$ and decreasing in $x$.
For $0 \leq \rho \leq 1$, let
\[
\xi(\rho) = \inf\left\{x\ge 0:
\frac{\rho}{x^{\beta}/{T r^{\beta}} + 1 - \rho} + I(\rho,B_0(x)) < 1 \right\}.
\]
Clearly, $\xi(\rho)$ is increasing in $\rho$.
Using this and (\ref{eq:key}), we get that $\mathbb{P}^0(\psi  > \rho)= \mathbb{P}^0(R_1>\xi(\rho))$.
From Slivnyak's Theorem, $\mathbb{P}^0(R_1>\xi(\rho))=\mathbb{P}(R_1>\xi(\rho))$.
This together with the formula for the probability that
the ball with radius $\xi(\rho)$ contains no point of a PPP of intensity $\lambda$
give:
\begin{theorem}If $S = B_0(R_1)$, then
$\mathbb{P}^0(\psi  > \rho) = \exp(-\lambda \pi \xi(\rho)^2)$
for $\rho<1$
and
$\mathbb{P}^0(\psi  = 1) = \exp(-\lambda \pi \xi(1)^2).$
\end{theorem}
To illustrate the last theorem, assume that for a~(small enough)
$\rho$, $\xi(\rho) = 0$. This means that all the nodes should attempt
with probabilities lower bounded by $\rho$ in each slot. Similarly,
assume that $\xi(1) = R^{\ast}$. Then all the nodes for whom the closest
neighboring receivers are at distances greater than $R^{\ast}$ should
transmit in every slot.
%
\subsection{Mean Utility}
\label{sec:mean-utility}
We now provide an analytical expression
for the mean optimized utility per unit space:
\begin{equation}
\label{eq:meanu}
\Theta^S= \lambda\left(\mathbb{E}^0[\log(\psi^S )]
+ \mathbb{E}^0[\log(q_0^S)]\right).
\end{equation}
Let $f(\cdot)$
denote the distribution of $\psi^S$ (see Section \ref{ss:mapdis}).
For all $t\in \R^2$,
let $f_t$ be the distribution of $\psi^S$ with an
extra receiver at $t$, i.e. the distribution of the solution of
\begin{eqnarray*}
\frac{1}{\psi}  = D(\psi,S)
+ \frac{\one(t\in S)}{1 + \frac{|t|^\beta}{|r|^\beta T} - \psi}
+\int\limits_{y \in \mathbb{R}^2\setminus S}\frac{\lambda{\rm d}y}{1+|y|^\beta/Tr^{\beta} - \psi}
\end{eqnarray*}
if it exists and  $\psi^S = 1$  otherwise.
Notice that $f(.)$ and $f_t(.)$
both depend on $S$.
For each $S$ studied
above, one gets an analytical expression for $f_t$ by the same technique
as in \S \ref{ss:mapdis}.
\begin{theorem}
\label{thm44}
For all $S$, $\Theta^S$ is given by:
\begin{align}
\label{eqn:mean-utility2}
\Theta^S = &
\lambda \int\limits_{[0,1]}
\log(u)f({\rm d}u) \\
          + & \lambda
\int\limits_{t\in \R^2}
\int\limits_{u\in [0,1]}
\log\left(1 - \frac{u }{1 + |t|^\beta/Tr^{\beta}} \right)
f_t({\rm d}u) {\rm d}t.
\!\! 
\nonumber
\end{align}
\end{theorem}
\begin{proof}
We have
\begin{eqnarray}
\mathbb{E}^0[\log(q_0^S)] & = &
\mathbb{E}^0\left[\sum_{j \neq 0} \log \left(1 - \frac{\psi^S }{1 + b_{0j}} \right)\right]\label{eq20}\\
&  &\hspace{-1.8cm}= \mathbb{E}^0\left[\sum_{j \neq 0} \log
\left(1 - \frac{\psi^S (\delta_{y_j}+ \sum_{i\ne j,0}
\delta_{y_i}) }{1 + \frac{|y_j|^\beta}{r^\beta T}} \right)\right]\label{eq21}\\
&  & \hspace{-1.8cm}= \lambda \int_{\R^2} \mathbb{E}^0\left[
 \log \left(1 - \frac{\psi^S (\delta_{t}+ \bar\Phi)
}{1 + \frac{|t|^\beta}{r^\beta T}} \right)\right]{\rm d}t\label{eq22}\\
&  &\hspace{-1.8cm}= \lambda \int_{\R^2}  \int_{[0,1]}
\log \left(1 - \frac{u}{1 + \frac{|t|^\beta}{r^\beta T}} \right)f_t({\rm d}u){\rm d}t\nonumber.
\end{eqnarray}
In these equations, we first used (\ref{eq:mastra}) in Appendix 
to get (\ref{eq20}), then stressed the fact that
$\psi^S$ is a function of $\bar \Phi$ by writing
$\psi^S=\psi^S(\sum_{i\ne 0} \delta_{y_i})$ in (\ref{eq21}),
and finally used the fact that $\bar \Phi$ is Poisson and
Campbell's formula to get (\ref{eq22}).
\end{proof}
\subsection{Convergence}
The aim of this section is to substantiate the claim that the
control schemes of Section \ref{sec:prop-fair-aloha} can
get arbitrarily close to the full information scheme.
We illustrate this for the case of the deterministic stopping sets
$B_0(R)$ by proving some continuity results of the main
performance metrics when $R$ tends to infinity.
When $R\to \infty$, we have
\begin{eqnarray*}
\sum_{{j \neq 0}\atop{y_j \notin B_0(R)}} \frac{1}{\frac{(|y_j|/r)^{\beta}}{  T}  + 1-\psi } \to 0,\quad
\int\limits_{R/r}^{\infty}\frac{s}{\frac{s^{\beta}}{  T} + 1 - \psi } {\rm d}s \to 0.
\end{eqnarray*}
It then follows from (\ref{eq:b0R}) and from standard
(deterministic) calculus arguments which are skipped here that, $\Pr^0$ a.s.,
$\lim_{R\to \infty} \psi^{B_0(R)}=\psi^{\R^2}$,
which proves the continuity of the MAP at infinity.
The main result of this section is:
\begin{lemma}
$\Theta^{B_0(R)}$ tends to $\Theta^{\R^2}$ when $R\to \infty$.
\end{lemma}
\begin{proof}
In view of (\ref{eq:meanu}) and (\ref{eq20}), this will hold
if we can interchange the a.s. limits w.r.t. $R$ and the expectations
in both terms of the R.H.S. of (\ref{eq:meanu}).
In order to do so, thanks to the dominated convergence theorem, it
is enough to show that each of the positive
random variables $-\log (\psi^{B_0(R)})$
and $-\log (q_0^{B_0(R)})$ is uniformly bounded from
above by some random variables with finite mean, where uniformity is w.r.t. $R$ large enough.
For $-\log (\psi^{B_0(R)})$, 
using the bound $-\log (\psi^{B_0(R)})\le 1/\psi^{B_0(R)}$, we see that
it is enough to prove the property for $1/\psi^{B_0(R)}$. But we get from
(\ref{eq:b0R}) and elementary monotonicity arguments that for all $R>1$,
$\frac 1 {\psi^{B_0(R)}}  \le  \frac {K_1} {1-\psi^{B_0(R)}} + A(R)$
with $K_1$ the number of receivers in $B_0(1)$ (excluding $y_0$) and
$$ A(R) = \sum_{j \neq 0:y_j\in B_0(R)\setminus B_0(1) }\frac{r^\beta}{T|y_j|^{\beta}} 
+ 2 \pi \lambda r^2 \int_{R/r}^{\infty}\frac{s^{1-\beta}}{  T} {\rm d}s.$$
This gives
$\frac 1 {\psi^{B_0(R)}}  \le 1+K_1 + A(R)\le 1+K_1 + A(1)$,
which is the announced upper bound.
The proof for $-\log (q_0^{B_0(R)})$ is similar and is skipped.
\end{proof}

\section{Numerical Evaluation and Simulation}
\label{ss:numerics}


\remove{
\subsection{Computation of the Integrals}

We used Matlab to evaluate the integrals. The infinite integral that shows up in
the expression of the Laplace transform (\ref{eq:lapl}) is handled
without truncation by Maple and Matlab.
The singularity at $w=0$ in the contour integrals
(\ref{eq:Parsev}) leveraging Parseval's theorem
is a false singularity and it is also
handled without further work by Matlab.
\remove{
The Matlab code is particularly efficient and is
used throughout the analytical evaluations described
below.
}
}
\paragraph*{Simulation Setting}
We consider a two dimensional square with side length $L$, and
$N$ nodes placed independently over this square according to the uniform
distribution; this corresponds to $\lambda = N/L^2$ in the stochastic
geometry model\footnote{A finite snapshot of a Poisson random process would contain a Poisson distributed number of nodes. However, for large $\lambda L^2$, the Poisson random variable with mean $\lambda L^2$ is highly concentrated around its mean. Thus we can use $\lambda L^2$ nodes for all the realizations in our simulation.}.
Each node has its receiver randomly located on the
unit circle around it, again as per the uniform
distribution. Thus $r_{ii} = 1$ for all $i$. We set $\alpha = 4$
and $T = 10$. To nullify the edge effect, we take into account
only the nodes falling in the $L/2 \times L/2$ square around the center while computing
various metrics. While all other parameters remain, we vary $L$ and $N$ for different simulations.
For each parameter set we calculate the average of the performance metric of interest over $1000$ independent
network realizations.
\paragraph*{Joint Validation of the Analysis and the Simulation}
\label{subsec:validation}
For illustration, we consider $S = \R^2$ and 
study the c.d.f. $f_t$ 
of the MAP $\psi^{\R^2}$~(see Section~\ref{sec:mean-utility}).
We set $L = 40$ and $N = 400$, which corresponds to $\lambda = 0.25$.
Figure~\ref{fig:map-cdf-receiver} shows the plots for $\lambda = 0.25$ and two values of $t$:
$t= 1$ and $t = 10$. 
\begin{figure}[h]
\centering
\includegraphics[width=0.8\linewidth]{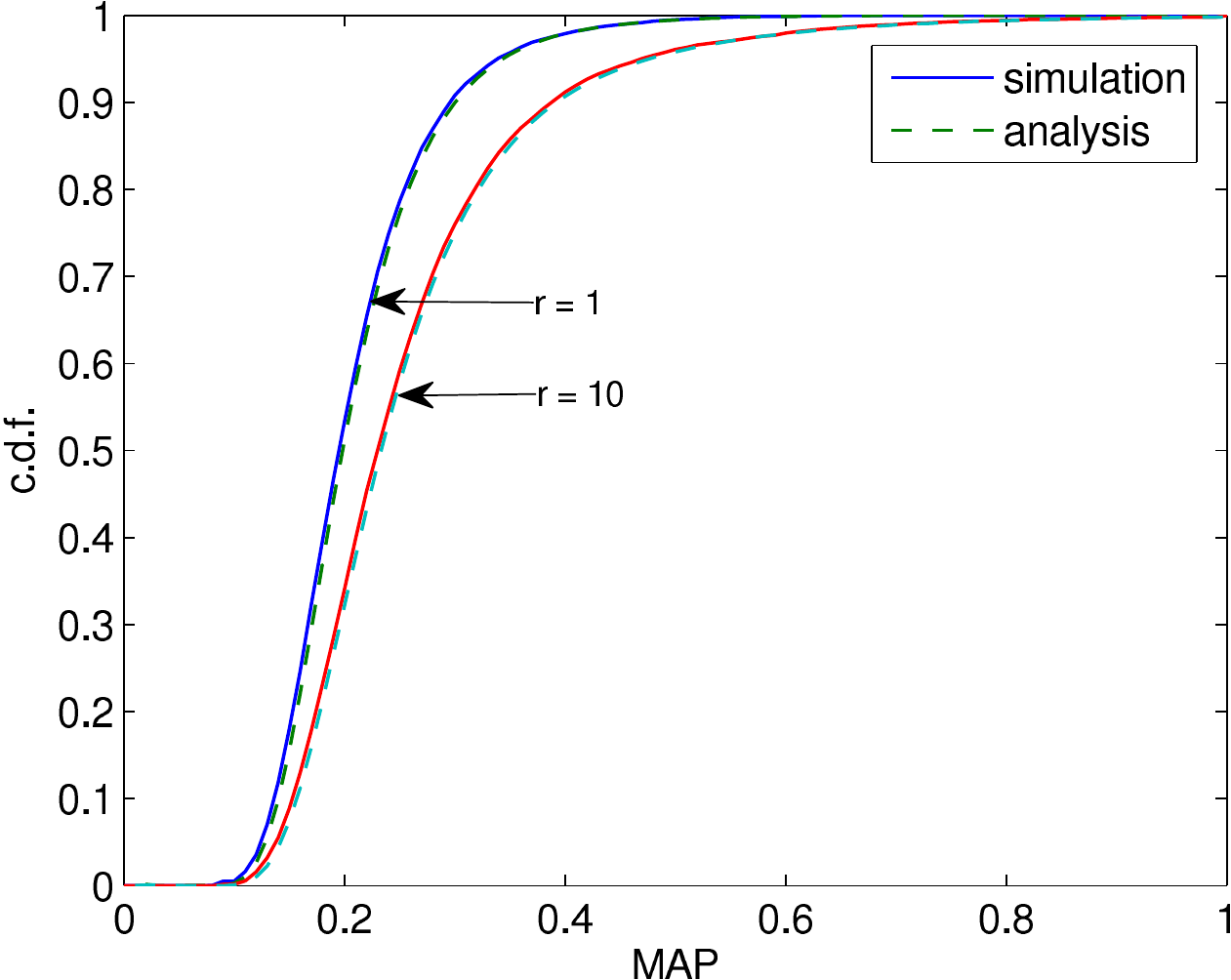}
\caption{C.d.f. $f_t$ of the MAP $\psi^{\R^2}$ with an extra receiver.}
\label{fig:map-cdf-receiver}
\end{figure}
First, observe that the stochastic geometry based
formula~(see Theorem~\ref{thm:map-distribution}) matches
simulation quite accurately.
As expected, $f_1(\psi^{\R^2}) \geq f_{10}(\psi^{\R^2})$ for all 
$\psi^{\R^2}$, i.e., the tagged node~(at origin) 
is more likely to be inactive for $t = 1$.

\paragraph*{Constant Stopping Sets}
\label{sec:performance}
We illustrate here the variation of performance w.r.t. information
on the case $S=B_0(R)$ through two sets of curves (both obtained
by simulation). The top picture
of Figure \ref{fig:mean-thput-log} quantifies the gains of the
{\em mean of the logarithm
of the typical node's throughput} as $R$ increases.
\begin{figure}[h]
\centering
\includegraphics[width=0.95\linewidth]{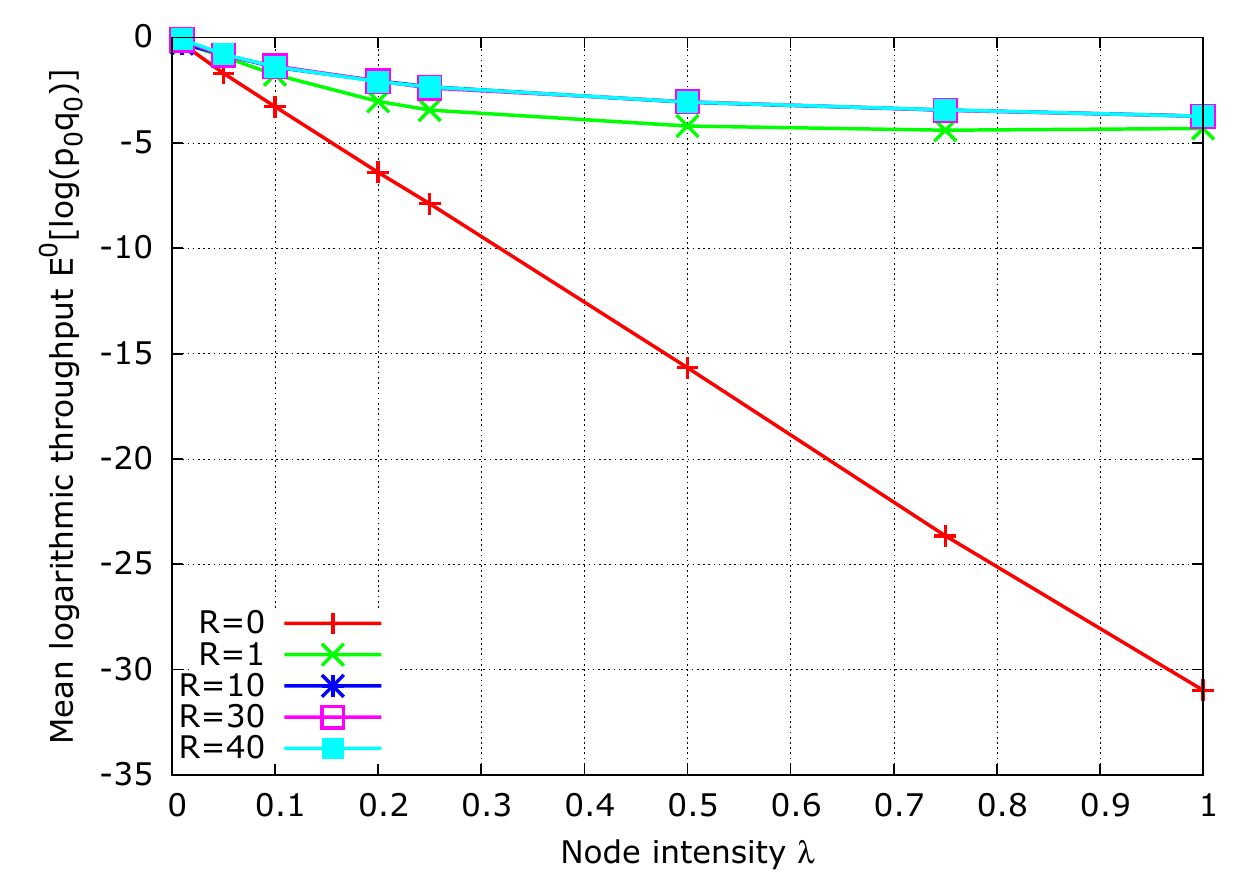}
\includegraphics[width=1\linewidth]{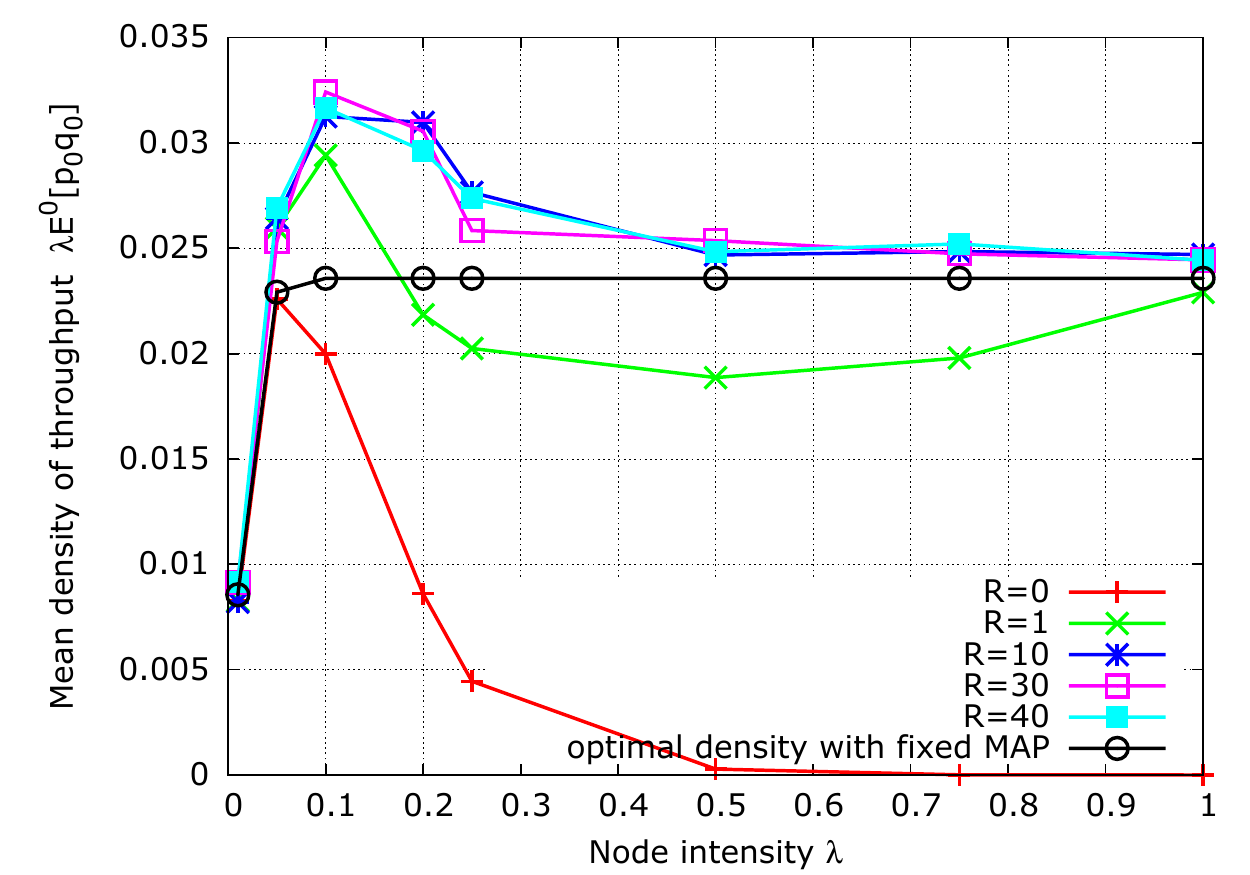}
\caption{Top: mean logarithmic throughput as a function of $\lambda$ for $S=B_0(R)$. Bottom: Density of throughput as a function of $\lambda$ for $S=B_0(R)$. The curve entitled ``optimum density with fixed MAP'' gives the result of the parametric
optimization of the density of throughput
performed in \cite{ctrltheory-wireless.baccelli-etal06aloha-multihop-wireless}
for comparison.}
\label{fig:mean-thput-log}
\end{figure}
In contrast, the bottom picture
of Figure \ref{fig:mean-thput-log}
plots the variations of the {\em mean aggregate throughput per unit space}
(referred to as density of throughput) as a function
of $\lambda$ when $R$ varies.
A first observation, in line with
\cite{ctrltheory-wireless.hsu-su11channel-aware-aloha},
is that an increase of information
does {\em not} necessarily lead to a better performance for
this last metric (there is no reason for this to hold anyway
as the schemes optimize the above logarithmic metric). 
This is however true for small $\lambda$, where
the absolute gains are quite substantial. 

\paragraph*{Random Stopping Sets}
In Figure~\ref{fig:agg-thput},
we plot the average aggregate throughput (sum of throughputs
throughout the simulation window)
under the optimal control $\pi^S$ for
$S = \R^2, B_0(R_1)$ and $\emptyset$,
and for $\lambda$ varying from $0.02$ to $1$.
This figure shows that knowing~(and accounting for)
the distance of the closest receiver only brings most of
the potential gains in terms of aggregate throughput\footnote{The scheme 
$S=B_0(R_1)$, or more precisely its variant $S=B_0(R_1)\cap B_0(R)$,
which is better in terms of implementation (see Section \ref{sec:impl})
hence provides a good candidate in this class.}. 
We also note that while the gains are large for
smaller node intensities, they diminish as $\lambda$ increases.
\remove{
\begin{figure}[h]
\centering
\includegraphics[width=0.6\linewidth]{map-distance1.pdf}
\caption{MAP as a function of distance to the closest receiver; $\lambda = 0. 25$.}
\label{fig:map-distance}
\end{figure}

\begin{figure}[h]
\centering
\includegraphics[width=0.95\linewidth]{map-distribution1.pdf}
\caption{C.d.f. of the MAP~(closest receiver case); $\lambda = 0. 25$.}
\label{fig:map-distribution}
\end{figure}
}

\section{Implementation}
\label{sec:impl}
The data required at transmitter 0 (the tagged transmitter) in order to solve 
the fixed point Equation (\ref{eqn:opt-p0}) have two components: (1)
the parameters $b_{0j}$, for all $y_j$ in the stopping set $S$, which show up
in the sum of the R.H.S., and
(2) the intensity parameter $\lambda$ which shows up in the integral of this R.H.S.
The estimation of $\lambda$ will not be discussed here.
We will concentrate below on the mechanisms through which the sequence $b_{0j}$, $y_j\in S$,
on which the adaptation is based,
can be estimated by the tagged transmitter. Assuming that $T$ and $r$ are known,
the knowledge of $b_{0j}$ is equivalent to that
of the path-loss $|X_0-y_j|^\beta$ between the tagged transmitter and receiver $j$.
We center the discussion on the case $S=B_0(R)$.

\begin{figure}[t!]
\centerline{\includegraphics[width=1\linewidth]{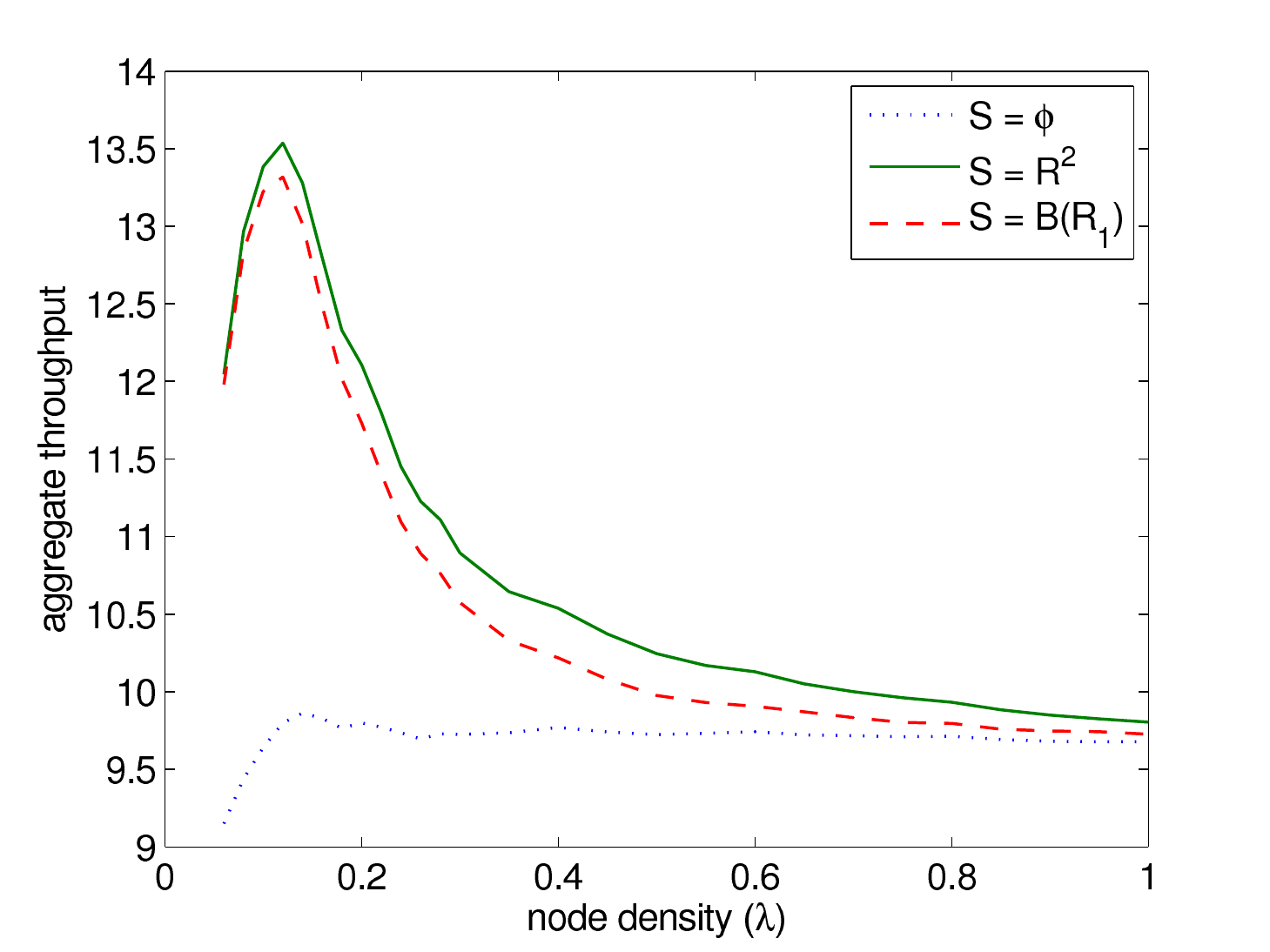}}
\caption{Aggregate throughput as a function of intensity.}
\label{fig:agg-thput}
\end{figure}

Our assumptions on the communication capabilities of the nodes are the usual ones for
an Aloha MAC setting:
(a) Each node (called {\em transmitter} or {\em receiver} until now) can actually both transmit and receive:
{\em transmitters} transmit packets and receive acknowledgements, whereas {\em receivers}
receive packets and transmit acknowledgements. 
(b) Each acknowledgement is transmitted with a known and constant power $P_a$.
(c) Each acknowledgement contains the identity of the transmitter-receiver pair.
In addition, we assume here that the tagged receiver can decode all acknowledgements stemming from receivers  
at distance less than $R$. This assumption is debatable in general but justified for moderate $R$.
\footnote{For instance, for $R$ of the order of $r$,
since the data rate required for acknowledgements is much lower than that for
payload packets, the SINR conditions for acknowledgements are easier
to meet than those for packets, so this assumption is acceptable.}

The main idea discussed below consists in using the acknowledgements that are
broadcast by all receivers to allow the tagged transmitter to estimate $|X_0-y_j|^\beta$
for all $y_j\in S$ (and in fact each transmitter to estimate in parallel
his path-loss to each receiver in his own stopping set). For this,
each transmitter decodes all acknowledgements 
and keeps a moving average of the powers at which acknowledgements stemming from node $j\in S$
were received.  
Since acknowledgements stemming from $j$ with $y_j\in S$ are assumed to
be successfully decoded with probability 1 and since the identity of receiver $j$ is assumed to
be contained in the acknowledgement, the tagged transmitter has the means to maintain
this moving average.  
The $k$-th element in the time series of received powers is $P_a G_{j0}^{(k)} |X_0-y_j|^\beta$,
with the sequence $\{G_{j0}^{(k)}\}_k$ made of i.i.d. exponential random variables
of mean 1 (because of the Rayleigh fading assumption).
By the strong law of large numbers, we have
$\lim_{n\to \infty} \frac 1 n \sum_{k=1,n} P_a G_{j0}^{(k)} |X_0-y_j|^\beta= P_a |X_0-y_j|^\beta \quad a.s.$.
So the tagged node can indeed estimate $|X_0-y_j|^\beta$ from his moving average
for all $y_j$ at a distance less than $R$, and hence $b_{0j}$ for all $y_j\in S$.

\section{Conclusion}
\hspace{-.2cm}
This paper determines in an analytic way the performance of a Poisson network
using a version of Aloha where each transmitter willingly adapts
his medium access probability based on local information
so as to reach a network wide proportional fairness.
To the best of our knowledge, this is the first
example of successful combination of stochastic geometry and
local adaptive protocol design aimed at optimizing a utility function
within this Aloha setting.
The particular case where the local information is limited
to the interference incurred by the closest receiver is shown to
be have some implementable variants and to bring most of the potential performance gain.
\appendix
We give here the proof of Theorem~\ref{t:opt-pi}.  From~\eqref{eqn:succ-prob},
\begin{eqnarray}\nonumber
\mathbb{E}^0[\log(p_0q_0)]
&= &\mathbb{E}^0[\log( \psi)+ \sum_{X_j \neq 0} \log \left(1 - \frac{p_j}{1 + b_{j0}} \right) ] \\
 &   &\hspace{-2cm}
=\mathbb{E}^0[\log( \psi)]
+ \mathbb{E}^0[\sum_{X_j \neq 0} \log \left(1 - \frac{ \psi}{1 + b_{0j}} \right) ]\,,
\label{eq:mastra}
\end{eqnarray}
where the second equality is due to 
the {\it mass transport principle}~(cf.~\cite[Page~65]{stochproc-wireless.baccelli-blaszczyszyn09stochastic-geometry-wireless-networks-2}).
Our goal is  to maximize the last expression in function 
of $\psi$ under the restriction~(\ref{e.pi-S}) for the given stopping set
$S=S(\Phi,\bar\Phi)$. For this, we split the sum in this expression into two terms
depending on whether  $y_j\in S$ or  $y_j\not\in S$.
By the  {\it strong Markov property};  cf \cite[Definition~1.16]{stochproc-wireless.baccelli-blaszczyszyn09stochastic-geometry-wireless-networks-1})
and Campbell's formula for the Poisson point process $\bar{\Phi}$, we
get that the second term of R.H.S. of (\ref{eq:mastra}) is equal to
$\E^0\Bigg[\int_{y \in \mathbb{R}^2\setminus S}\log\left(1 - \frac{\psi}{1 +
    |y|^\beta/Tr^{\beta}} \right){\rm d}y \Bigg] \,.$
We thus conclude that ${\bf PF}^S$ is equivalent to 
the maximization of the following expectation 
in the considered class of functions $\psi(\cdot)$ 
\begin{align}\label{e.expectation-to-maximize}
\E^0\biggl[\log(\psi) & + \sum_{y_j \in S,j\not=0} \log \left(1 -
  \frac{\psi}{1 + b_{0j}} \right) \\ \nonumber
& + \lambda\int_{y \in \mathbb{R}^2\setminus S}\log\left(1 - \frac{\psi}{1 + |y|^\beta/Tr^{\beta}} \right) {\rm d}y\biggr]\,.
\end{align}
It is obvious that maximizing the expression under the expectation for
any given realization of $\bar\Phi$ (it does not depend on
$\Phi$) we will obtain a solution of  ${\bf PF}^S$.
The derivative of this expression with respect to $\psi$ is equal to 
$$\frac{1}{\psi}-\sum_{y_j\in S, j\not=0}\frac{1}{1 + b_{0j} - \psi}
-\int_{y \in \mathbb{R}^2\setminus S}\frac{\lambda{\rm d}y}{1+|y|^\beta/Tr^{\beta} - \psi}\,.$$
It is continuous and  decreasing in $\psi$  over $[0, 1]$.
This proves that the MAC policy $\psi^S$ 
is a solution of   ${\bf PF}^S$. 
In order to prove that $\E^0[\psi^S q_0]>-\infty$ observe that  
the expectation~(\ref{e.expectation-to-maximize}) is finite
($>-\infty$) for a constant MAC policy $\psi=\text{\em const}\in(0,1)$. 

Finally, in order to prove the uniqueness,
suppose  
that for some other MAC policy~$\psi'(\cdot)$  the value of the
expectation~(\ref{e.expectation-to-maximize}) is equal to that 
obtained for $\psi^S$; i.e., 
$\E^0[H(\psi^S(\bar\Phi))]=\E^0[H(\psi'(\bar\Phi))]$,
where $H(\cdot)$ is the expression under the expectation
in~(\ref{e.expectation-to-maximize})  as a function of the MAC policy.
Then we have
$\max(H(\psi^S(\cdot),H(\psi'(\cdot))-H(\psi'(\cdot))\ge0$
and 
\begin{align*}
\E^0&[\max(H(\psi^S(\bar\Phi)),H(\psi'(\bar\Phi))-H(\psi'(\bar\Phi))]\\
&\hspace{-.4cm}=\E^0[\max(H(\psi^S(\bar\Phi)),H(\psi'(\bar\Phi))]-\E^0[H(\psi'(\bar\Phi))]\\
&\hspace{-.4cm}\ge
\max(\E^0[H(\psi^S(\bar\Phi))],\E^0[H(\psi'(\bar\Phi))])-\E^0[H(\psi'(\bar\Phi))]=0
\end{align*}  
since 
$\E^0[H(\psi'(\bar\Phi))]=\E^0[H(\psi^S(\bar\Phi))]>-\infty$.
Consequently
$\max(H(\psi^S(\bar\Phi)),H(\psi'(\bar\Phi))=H(\psi'(\bar\Phi))=
H(\psi^S(\bar\Phi)$
for almost all realizations of $\bar\Phi$. The fact that $H(\cdot)$
has one maximum in $[0,1]$ (attained for $\psi^S$) allows one to conclude
that  $\psi'(\bar\Phi)=\psi^S(\bar\Phi)$
for almost all realizations of $\bar\Phi$.
\bibliographystyle{abbrv}\pdfbookmark[1]{References}{references} 
{\small

\end{document}